\documentclass[12pt]{article}

\usepackage[letterpaper,margin=1.25in]{geometry}
\usepackage[T1]{fontenc}
\usepackage{mathpazo}
\usepackage{amsmath,amssymb,amsthm,mathtools}
\allowdisplaybreaks
\numberwithin{equation}{section}

\usepackage{enumitem}
\usepackage{xcolor}
\usepackage[colorlinks=true,
            linkcolor={black!65!blue},
            citecolor={black!55!green},
            urlcolor={black!45!blue},
            linktoc=all]{hyperref}

\usepackage[authoryear,round]{natbib}

\usepackage{titlesec}
\titleformat{\section}{\normalfont\large\bfseries}{\thesection}{0.6em}{}
\titleformat{\subsection}{\normalfont\normalsize\bfseries}{\thesubsection}{0.6em}{}
\titlespacing*{\section}{0pt}{1.5\baselineskip}{0.5\baselineskip}
\titlespacing*{\subsection}{0pt}{1.1\baselineskip}{0.3\baselineskip}

\theoremstyle{plain}
\newtheorem{theorem}{Theorem}[section]
\newtheorem{lemma}[theorem]{Lemma}
\newtheorem{proposition}[theorem]{Proposition}
\newtheorem{corollary}[theorem]{Corollary}
\theoremstyle{definition}
\newtheorem{definition}[theorem]{Definition}
\newtheorem{axiom}[theorem]{Axiom}
\newtheorem{remark}[theorem]{Remark}
\newtheorem{example}[theorem]{Example}

\title{Strategy-Proof Set-Valued Choice on Single-Dipped Domains}

\author{Arkarup Basu Mallik \footnote{Economics Research Unit, Indian Statistical Institute, Kolkata. Email: arkarupbasumallik97@gmail.com}\, Mihir Bhattacharya\footnote{Department of Economics, Ashoka University, Rajiv Gandhi Education City, Rai, Sonipat,  Haryana, 131029, India. Email: mihir.bhattacharya@ashoka.edu.in} \,
Ojasvi Khare\footnote{Shiv Nadar University, Department of Economics, Gautam Buddha Nagar, Uttar Pradesh, 201314. Email: ojasvi.khare@snu.edu.in} }

\date{}

\begin{document}

\maketitle           

\begin{abstract}
We study strategy-proof social choice where agents have single-dipped preferences and outcomes may be sets of alternatives. While point-valued strategy-proof rules select an endpoint, we ask which outcomes survive when sets are compared through responsive extensions. Under the standard responsive extension, preferences over fixed-length intervals remain single-dipped, so the endpoint restriction persists and anonymity yields a single-threshold rule. Under responsive dominance, the incomplete order shared by all responsive extensions, an interior interval can survive because the two extreme intervals need not be comparable. For arbitrary sets, requiring every chosen alternative to be most preferred by some agent yields  voting-by-committees rules, with one committee per endpoint; anonymous rules have a two-threshold quota form. The characterization holds on the entire endpoint-peaked domain, and every such rule is group-strategy-proof.

\end{abstract}

\section{Introduction}\label{sec:intro}
Consider the location of an undesirable public facility on a line---a landfill, a prison, or a transmitter---so that each agent prefers the facility to be as far as possible from a single least preferred point. Preferences of this kind are \emph{single-dipped}: each has a worst alternative, the dip, and preference increases with distance from it. Consequently, the most preferred alternative is always one of the two endpoints of the line. Single-dipped preferences are the reflection of the single-peaked preferences of \cite{black1948} and \cite{moulin1980}, but the two domains behave very differently under strategy-proofness. On single-peaked domains, strategy-proof rules are generalized medians and may select interior alternatives (\cite{moulin1980}). On the full single-dipped domain, the range of a unanimous strategy-proof rule is confined to the two endpoints. A further distinctive feature of the domain is that strategy-proofness and group strategy-proofness coincide under the relevant assumptions (\cite{manjunath2014} and \cite{barbera2012}). \cite{manjunath2014} describes the resulting rules through a collection of decisive coalitions that determines which endpoint is chosen. More generally, \cite{peremansstorcken1999} show that on any subdomain of the single-dipped domain, the range of a strategy-proof rule has cardinality at most \(2^n\); on the full domain this bound shrinks to two alternatives.

This paper asks how this endpoint result changes when the social choice rule may output a \emph{set} of alternatives rather than a single alternative. Set-valued outcomes arise naturally when a planner selects a region rather than a point, produces a shortlist of candidate locations, or deliberately leaves a final choice unresolved. Such outcomes require a preference over sets.

We use responsive extensions of \cite{gardenfors1976} and \cite{kelly1977} to compare sets. A responsive extension respects the underlying preference when equal-sized sets differ only by replacing one alternative with another, and satisfies G\"ardenfors monotonicity when an alternative is added to a set. A given preference over alternatives generally admits many such extensions. Their intersection, which we call \emph{responsive dominance}, is therefore only a partial order.

The contrast between a fixed responsive extension and responsive dominance runs through the paper. Under a suitable common responsive extension, preferences over fixed-length intervals inherit the single-dipped structure of the underlying preferences, and the classical endpoint restriction survives. Under responsive dominance, by contrast, the relevant comparisons may be incomplete, and an interior interval can survive because the two extreme intervals need not be comparable. Thus, what drives the difference is not merely that outcomes are set-valued, but how preferences over sets are extended.

We first develop the fixed-length interval result as a benchmark. When agents evaluate intervals through a common responsive extension satisfying suitable peak-consistency and richness properties, the induced preferences over the ordered intervals form the full single-dipped domain on that chain; the leximax extension provides a canonical example. Every unanimous and strategy-proof interval rule therefore has range contained in the two extreme intervals, and anonymity yields a single-threshold rule (Theorem~\ref{thm:interval}). By contrast, under responsive dominance an interior interval can be Pareto efficient and immune to unilateral manipulation (Example~\ref{ex:interior}), precisely because the two extreme intervals may be incomparable to an agent whose dip lies in the interior (Example~\ref{ex:leximaxmin}). The point-valued endpoint argument therefore does not automatically extend to set-valued outcomes under dominance.

Our main results turn to arbitrary nonempty subsets of the line. Sets are compared by responsive dominance and, in addition to anonymity and strategy-proofness, we impose one selection requirement: an alternative may be chosen only if it is the most preferred alternative of some agent. On the single-dipped domain, this requirement has a particularly sharp implication. Since every agent's most preferred alternative is an endpoint, every outcome can contain only one endpoint, the other endpoint, or both. Although responsive dominance is generally incomplete, it is complete over these three possible outcomes, and an agent's ranking among them depends only on which endpoint she prefers. The problem therefore reduces to social choice over a three-element chain with two possible preference types.

Our main characterization (Theorem~\ref{thm:committee}) shows that top-selection and strategy-proofness are equivalent to a \emph{voting-by-committees} representation in the sense of \cite{barbera1991}. Two monotone committees determine, respectively, whether $x_1$ and $x_m$ belong to the outcome, subject to a covering condition ensuring that the outcome is nonempty and boundary conditions ensuring top-selection. Thus the committee structure is not imposed as a voting primitive: it is derived from strategy-proofness and top-selection.

This two-committee structure is the set-valued counterpart of the single decisive family that governs point-valued choice. In the point-valued problem, selecting one endpoint is the same act as excluding the other. When the outcome may be a set, however, admitting $x_1$ and admitting $x_m$ become separate membership decisions governed by the committees $W_1$ and $W_m$. The covering condition links them only through non-emptiness, and the compromise outcome $\{x_1,x_m\}$ arises  when both committees are decisive. Thus the point-valued dichotomy becomes a trichotomy: the two endpoint alternatives are unchanged, but set-valuedness allows them to be combined into a third outcome with no point-valued counterpart. Under anonymity, the single threshold of \cite{manjunath2014} correspondingly becomes a pair $q_1\leq q_2$ (Corollary~\ref{cor:quota}), with $q_1=q_2$ recovering the single-valued rule as the case in which the two-element outcome is never selected.

The characterization also has two further implications. First, the rules are automatically group-strategy-proof (Proposition~\ref{prop:gsp}). This follows from a general monotonicity lemma: on a finite chain on which each agent has one of two opposite preferences, strategy-proofness is equivalent to monotonicity, and monotonicity implies group-strategy-proofness (Lemma~\ref{lem:chain}). The same observation drives both the interval benchmark and the committee characterization. Second, single-dippedness is stronger than necessary for the main result. The characterization continues to hold on the entire endpoint-peaked domain, which is strictly larger than the single-dipped domain (Proposition~\ref{prop:maximal}). What matters for the committee result is therefore not the full single-dipped structure, but the location of the peaks at the two endpoints.

Finally, top-selection cannot simply be replaced by Pareto efficiency. An interior alternative may be Pareto efficient even on the single-dipped domain (Remark~\ref{rem:eff}). Pareto efficiency therefore does not prevent the rule from selecting alternatives that no agent ranks first, whereas top-selection rules them out by construction. This distinction clarifies the role of the additional axiom: the three-outcome range, and hence the voting-by-committees structure, emerges because top-selection eliminates the interior alternatives.

The results thus separate three effects of allowing set-valued outcomes. Under a suitable complete responsive extension, the classical endpoint restriction survives for fixed-length intervals. Under responsive dominance, incompleteness permits interior intervals to survive. Once top-selection restricts arbitrary outcomes to the endpoint-generated three-element range, however, dominance becomes complete again and strategy-proofness recovers a clean voting structure. The paper therefore identifies both the source of the failure of the classical endpoint argument and the additional selection requirement under which a tractable characterization is restored.

\paragraph{Related literature.} The domain-restriction response to the impossibility of \cite{gibbard1973} and \cite{satterthwaite1975} begins on single-peaked domains with \cite{black1948} and \cite{moulin1980}, where strategy-proof rules are generalized medians. Single-dipped domains provide a strikingly different restricted-domain model. On arbitrary subdomains, strategy-proofness already imposes strong restrictions on the size of the range; on the full domain of single-dipped preferences, the range contains at most two alternatives (\cite{peremansstorcken1999}, \cite{manjunath2014}, \cite{barbera2012}). Thus the median rules of the single-peaked world are replaced by endpoint rules governed by binary decisiveness. We ask how this point-valued endpoint structure changes once outcomes may be sets, and how the single decisive family is transformed when admitting the two endpoints need no longer be mutually exclusive decisions.

Beyond the public-good location model, single-dipped preferences have also been studied in other environments, including the division of a private good, where strategy-proofness can force dictatorship (\cite{klaus1997}, \cite{klaus2001}), and object reallocation (\cite{tamura2023}). Mixed domains combining single-peaked and single-dipped preferences have likewise been characterized: \cite{alcaldeunzu2024} describe strategy-proof rules when the planner knows each agent's preference type. These models remain point-valued; our contribution is to move to set-valued outcomes on the single-dipped domain itself.

Closest to our benchmark is the literature on \emph{interval} social choice, in which the outcome is a contiguous block of alternatives---a segment of land, a stretch of a street---rather than a point. \cite{miyagawa2001} locates a facility as a segment on a street, and \cite{klausprotopapas2020} and \cite{bhattacharyakhare2024} study single-peaked preferences extended to fixed-length intervals by responsiveness, characterizing the strategy-proof rules as generalized-median (interval) rules, which select interior intervals. Our interval benchmark is the single-dipped counterpart of this single-peaked model: in the same responsiveness-based framework, strategy-proofness forces the outcome to the two \emph{extreme} intervals rather than to a median (Theorem~\ref{thm:interval})---the dual of the single-peaked result---and the interior-interval possibility we identify (Example~\ref{ex:interior}) is specific to the partiality of responsive dominance, which has no single-peaked analogue.

More broadly, the strategy-proofness of social choice \emph{correspondences}---rules returning sets of alternatives---has been studied since the foundational work of \cite{gardenfors1976} and \cite{kelly1977}, with subsequent characterizations under alternative extensions of preferences over sets. In much of that literature, a chosen set represents indeterminacy, a tie, or a collection of alternatives pending further resolution. Here, by contrast, the set is itself the object of choice, is evaluated directly through responsive dominance, and is tied to the agents' most preferred alternatives through top-selection.

Our rules are group-strategy-proof (Proposition~\ref{prop:gsp}), connecting the paper to the literature on domains where individual and group strategy-proofness coincide (\cite{peremansstorcken1999}, \cite{lebretonzaporozhets2009}, \cite{barberabergamoreno2010}). Once top-selection reduces the range to a short chain, the problem also connects to the study of strategy-proof choice over binary and small ranges (\cite{larssonsvensson2006}, \cite{manjunath2012}, \cite{barbera2012binary}). The difference is that our three outcomes arise endogenously from set-valued choice: the two endpoint alternatives may be selected separately or jointly.

Finally, the committee representation in our characterization is the voting-by-committees form of \citet{barbera1991}, with the two endpoints as the objects whose membership in the outcome is determined. Constrained and set-valued committee problems on single-peaked domains are studied by \citet{barbera2001}. Our contribution is to show that, in the single-dipped setting, responsive dominance and top-selection together reduce arbitrary set-valued choice to two linked endpoint-membership decisions, yielding a committee representation derived from strategy-proofness rather than imposed as a voting primitive.

The paper proceeds as follows. Section~\ref{sec:model} sets up the model and records the monotone-selection lemma. Section~\ref{sec:intervals} treats the interval benchmark, and Section~\ref{sec:faithful} the general set-valued problem. Section~\ref{sec:conclusion} concludes.

\section{Model}\label{sec:model}

Let $N=\{1,\ldots,n\}$ be the set of agents and $X=\{x_1,\ldots,x_m\}$ the set of alternatives, ordered $x_1<\cdots<x_m$, with $m\geq 3$.

\begin{definition}[Single-dipped preference]
A \textbf{strict preference} $P_i$ on $X$ is a strict linear order on $X$. It is \textbf{single-dipped} if there is $d(P_i)\in X$, called the \emph{dip}, such that for all $x,y\in X$,
\begin{align}
x<y\leq d(P_i) &\implies xP_i y, \label{eq:left}\\
d(P_i)\leq x<y &\implies yP_i x. \label{eq:right}
\end{align}
\end{definition}

Thus, along either side of the dip, preference improves as one moves away from it. Let $\tau(P_i)$ denote the unique most preferred alternative under $P_i$. The top alternative must be an endpoint. Indeed, if $d(P_i)=x_k$, then for every $1<j\leq k$, \eqref{eq:left} gives $x_1P_ix_j$, while for every $k\leq j<m$, \eqref{eq:right} gives $x_mP_ix_j$. Hence every interior alternative is beaten by an endpoint, and $\tau(P_i)\in\{x_1,x_m\}$.

Let $\mathcal D$ be the set of single-dipped preferences and let $P\in\mathcal D^n$ be a profile. For an agent $i$ and report $P_i'$, let $(P_i',P_{-i})$ denote the profile obtained from $P$ by replacing $P_i$ with $P_i'$, where $P_{-i}=(P_j)_{j\neq i}$. Similarly, for a coalition $S\subseteq N$ and reports $P_S'=(P_i')_{i\in S}$, let $(P_S',P_{-S})$ denote the profile obtained by replacing $P_i$ with $P_i'$ for each $i\in S$.

A \textbf{social choice function} is a map
\[
f:\mathcal D^n\to\mathcal A,
\]
where $\mathcal A\subseteq 2^X\setminus\{\emptyset\}$ is a family of feasible nonempty outcomes. Below, $\mathcal A$ will either be a family of intervals or the full family of nonempty subsets of $X$. To compare outcomes, we extend preferences over alternatives to preferences over sets.

\begin{definition}[Responsive extension]\label{def:respext}
Fix $P_i\in\mathcal D$. A \textbf{responsive extension} of $P_i$ is a strict linear order $\succ_i$ on the family of feasible nonempty outcomes $\mathcal A$ satisfying the following conditions.

\smallskip

\noindent\textbf{Responsiveness:} For $A,B\in\mathcal A$ of equal cardinality such that $A\setminus B=\{x\}$ and $B\setminus A=\{y\}$,

$$
A\succ_i B\quad\Longleftrightarrow\quad xP_i y.
$$

\noindent\textbf{G\"ardenfors monotonicity:} For $A\in\mathcal A$ and $z\notin A$, whenever $A\cup{z}\in\mathcal A$,

$$
zP_i y\ \text{for all }y\in A
\quad\Longrightarrow\quad
A\cup\{z\}\succ_i A,
$$

while

$$
yP_i z\ \text{for all }y\in A
\quad\Longrightarrow\quad
A\succ_i A\cup\{z\}.
$$

Write $A\succ_i^R B$ (\textbf{strict responsive dominance}) if

$$
A\succ_i B
$$

under every responsive extension of $P_i$, and write $A\succeq_i^R B$ if either $A=B$ or $A\succ_i^R B$. The existence of responsive extensions is standard (\cite{gardenfors1976}, \cite{kelly1977}). Responsive dominance is a partial order, obtained by retaining only those comparisons on which all responsive extensions agree.
\end{definition}

Two responsive extensions will recur below and are worth naming. Both compare sets of equal cardinality by ordering their elements according to $P_i$, differing only in whether the comparison begins with the best or the worst element.

\begin{definition}[Leximax and leximin]\label{def:leximaxmin}
Fix $P_i$ and let $A,B\in\mathcal A$ have equal cardinality $k$. Write their elements in decreasing order under $P_i$ as

$$
a_1P_i\cdots P_i a_k
\qquad\text{and}\qquad
b_1P_i\cdots P_i b_k.
$$

The \textbf{leximax extension} $\succ_i^{\mathrm{lmax}}$ compares the two lists from the top:

$$
A\succ_i^{\mathrm{lmax}}B
$$

if $a_tP_i b_t$ at the smallest index $t$ for which $a_t\neq b_t$.

The \textbf{leximin extension} $\succ_i^{\mathrm{lmin}}$ compares them from the bottom:

$$
A\succ_i^{\mathrm{lmin}}B
$$

if $a_tP_i b_t$ at the largest index $t$ for which $a_t\neq b_t$.
\end{definition}

For sets of equal cardinality, both rules satisfy responsiveness: if $A$ and $B$ differ in  one element, their ordered lists differ in  one position, and both rules compare the two differing elements according to $P_i$. They can nonetheless disagree, and whenever they do, responsive dominance is silent.

\begin{example}\label{ex:leximaxmin}
Let $m=4$ and suppose $P_i:\quad x_1\,P_i\,x_2\,P_i\,x_3\,P_i\,x_4.$ Compare $A=\{x_1,x_4\}
\qquad\text{and}\qquad
B=\{x_2,x_3\},$ whose $P_i$-ordered lists are $(x_1,x_4)$ and $(x_2,x_3)$, respectively. Leximax first compares the best elements, $x_1$ and $x_2$. Since $x_1P_ix_2$, it ranks $A\succ_i^{\mathrm{lmax}}B.$ Leximin first compares the worst elements, $x_4$ and $x_3$. Since $x_3P_ix_4$, it ranks $B\succ_i^{\mathrm{lmin}}A.$

Thus the two responsive extensions disagree. Consequently, neither $A\succ_i^R B$ nor $B\succ_i^R A$ holds: under responsive dominance, $A$ and $B$ are \emph{incomparable}. This incompleteness of responsive dominance allows an interior interval to survive in Example~\ref{ex:interior}.
\end{example}

We use throughout the standard requirements of Pareto efficiency, strategy-proofness, and anonymity and, in Proposition~\ref{prop:gsp}, the stronger requirement of group-strategy-proofness. Let $\succ_i^{*}$ denote agent $i$'s \textbf{operative preference relation} over feasible outcomes, with weak counterpart $\succeq_i^{*}$. The operative relation is specified separately in each environment. When outcomes are compared by responsive dominance, $\succ_i^{*}=\succ_i^R$; when a particular responsive extension is fixed, $\succ_i^{*}=\succ_i$. Thus, $\succ_i^{*}=\succ_i^R$ in the responsive-dominance part of Section~\ref{sec:intervals} and throughout Section~\ref{sec:faithful}, whereas $\succ_i^{*}$ is a fixed responsive extension in the relevant parts of Section~\ref{sec:intervals}.

\begin{axiom}[Pareto efficiency]\label{ax:pe}
A social choice function $f$ is \textbf{Pareto efficient} if, for every profile $P$, there is no feasible outcome $b$ such that

$$
b\succ_i^{*}f(P)
\qquad\text{for every }i\in N.
$$

\end{axiom}

\begin{axiom}[Strategy-proofness]\label{ax:sp}
A social choice function $f$ is \textbf{strategy-proof} if, for every profile $P$, every agent $i\in N$, and every report $P_i'\in\mathcal D$,

$$
f(P_i',P_{-i})\not\succ_i^{*}f(P).
$$

\end{axiom}

\begin{axiom}[Group-strategy-proofness]\label{ax:gsp}
A social choice function $f$ is \textbf{group-strategy-proof} if there does not exist a profile $P$, a coalition $S\subseteq N$, and reports $P_S'$ such that

$$
f(P_S',P_{-S})\succeq_i^{*}f(P)
\quad\text{for every }i\in S,
$$

and

$$
f(P_S',P_{-S})\succ_i^{*}f(P)
\quad\text{for some }i\in S.
$$

\end{axiom}

\begin{axiom}[Anonymity]\label{ax:anon}
A social choice function $f$ is \textbf{anonymous} if, for every profile $P=(P_1,\ldots,P_n)$ and every permutation $\pi$ of $N$,

$$
f(P_{\pi(1)},\ldots,P_{\pi(n)})
=
f(P_1,\ldots,P_n).
$$

\end{axiom}

\paragraph{Notation.}
We collect here notation used repeatedly below. For a strict preference $P_i$ on $X$, $\tau(P_i)$ denotes its top alternative and $d(P_i)$ its dip. The single-dipped domain is denoted by $\mathcal D$, and

$$
\mathcal E=\{P_i:\tau(P_i)\in\{x_1,x_m\}\}\supseteq\mathcal D
$$

denotes the \emph{endpoint-peaked} domain: it requires only that the peak be an endpoint and imposes nothing on the ranking below it, whereas single-dippedness further fixes that ranking through the dip. The inclusion is strict; for example, with $m=4$ the preference $x_1\,x_3\,x_2\,x_4$ is endpoint-peaked but not single-dipped, since single-dippedness with peak $x_1$ would require $x_1\,x_2\,x_3\,x_4$. Preferences over outcomes are written $\succ_i$, with weak counterpart $\succeq_i$.

A \emph{chain order} on a family of outcomes is a fixed strict linear order, denoted by $<_L$, with reflexive closure $\leq_L$. The chain order is exogenous and should not be confused with an agent's preference. In Section~\ref{sec:intervals},

$$
I_1<_L\cdots<_L I_M
$$

is the chain of intervals of length $\ell$, and $\rho(P_i)\in\{I_1,I_M\}$ denotes the induced top interval. The relations $\succ_i^R$ and $\succeq_i^R$ denote strict and weak responsive dominance, respectively. In Section~\ref{sec:faithful},
\[
N_1(P)=\{i:\tau(P_i)=x_1\},\qquad N_m(P)=N\setminus N_1(P),
\]
are the two endpoint-peaked coalitions, and $n_m(P)=|N_m(P)|$.

Finally, $\mathbf 1[\cdot]$ denotes the indicator function, equal to $1$ when its argument holds and $0$ otherwise.

We next relate strategy-proofness to an appropriate monotonicity property of social choice rules on the ordered family of intervals. Informally, monotonicity requires that when an agent changes her report in the direction of higher intervals, the selected outcome cannot move in the opposite direction, and symmetrically for a change in the direction of lower intervals. In the ordered environment considered below, this condition is equivalent to strategy-proofness. We now make this equivalence precise.

The same argument rules out coalitional manipulation. If a coalition moves the outcome upward in the chain order, any member whose operative preference ranks lower outcomes more highly is strictly worse off. Thus, a profitable deviation can involve only agents who prefer higher outcomes. But for such agents, the deviation changes their reports in the direction opposite to the resulting movement of the outcome, which monotonicity rules out. The case of a downward movement is symmetric.

The following lemma formalizes this argument.

\begin{lemma}[Monotone selection on a chain]\label{lem:chain}
Let $(\mathcal I,<_L)$ be the chain
\[
I_1<_L I_2<_L\cdots<_L I_M,
\]
with reflexive closure $\leq_L$. Each agent has one of two types, $\downarrow$ and $\uparrow$, determined by the direction in which her operative preference $\succ_i^{*}$ runs along the chain: type $\downarrow$ prefers outcomes against the order $<_L$, while type $\uparrow$ prefers them along it. Thus, for $A,B\in\mathcal I$,
\[
A\succ_i^{*}B\iff A<_L B\quad(\text{type }\downarrow),
\qquad
A\succ_i^{*}B\iff B<_L A\quad(\text{type }\uparrow),
\]
with weak counterpart $\succeq_i^{*}$ in each case. For a rule $g:\{\downarrow,\uparrow\}^n\to\mathcal I$, the following are equivalent:
\begin{enumerate}[label=(\roman*)]
\item $g$ is strategy-proof: for every type profile $t$, every agent $i$, and every report $t_i'$,
\[
g(t_i',t_{-i})\not\succ_i^{*}g(t);
\]
\item $g$ is monotone: for every $i$ and every $t_{-i}$,
\[
g(\downarrow,t_{-i})\leq_L g(\uparrow,t_{-i}).
\]
\end{enumerate}
Moreover, every monotone rule is group-strategy-proof.
\end{lemma}

\begin{proof}
(i)$\iff$(ii). Suppose first that $g$ is monotone, and fix $i$ and $t_{-i}$. If $i$ has type $\downarrow$, then
\[
g(\downarrow,t_{-i})\leq_L g(\uparrow,t_{-i}),
\]
so $g(\downarrow,t_{-i})\succeq_i^{*}g(\uparrow,t_{-i})$; hence reporting $\uparrow$ cannot benefit a type-$\downarrow$ agent. The same inequality gives $g(\uparrow,t_{-i})\succeq_i^{*}g(\downarrow,t_{-i})$ for a type-$\uparrow$ agent, so reporting $\downarrow$ cannot benefit her either. Thus $g$ is strategy-proof. Conversely, if $g$ is not monotone, then for some $i$ and $t_{-i}$,
\[
g(\uparrow,t_{-i})<_L g(\downarrow,t_{-i}),
\]
so $g(\uparrow,t_{-i})\succ_i^{*}g(\downarrow,t_{-i})$ for a type-$\downarrow$ agent, who can therefore profitably report $\uparrow$. Hence $g$ is not strategy-proof.

Group-strategy-proofness. Suppose a coalition $S$ can profitably deviate from a truthful type profile $t$ to $t'$, with $t_j'=t_j$ for every $j\notin S$. The two outcomes must differ. Suppose first that
\[
g(t)<_L g(t').
\]
Then $g(t)\succ_i^{*}g(t')$ for every type-$\downarrow$ agent, so no type-$\downarrow$ member of $S$ can be weakly better off; hence every member of $S$ has type $\uparrow$ at $t$. The deviation therefore only switches members from $\uparrow$ to $\downarrow$, and by monotonicity each such switch weakly lowers the outcome, giving $g(t')\leq_L g(t)$, a contradiction. The case $g(t')<_L g(t)$ is symmetric, with the roles of $\downarrow$ and $\uparrow$ interchanged. Thus no profitable coalitional deviation exists, and $g$ is group-strategy-proof.
\end{proof}

\section{Fixed-length intervals}\label{sec:intervals}

We begin by restricting attention to intervals of a common fixed length, so that the relevant outcomes admit a natural chain order by location.

Fix a length $\ell\in\{2,\ldots,m-1\}$. For $1\leq j\leq M:=m-\ell+1$, let
\[
I_j=\{x_j,\ldots,x_{j+\ell-1}\},
\]
and let
\[
\mathcal I=\{I_1,\ldots,I_M\},
\qquad
I_1<_L\cdots<_L I_M,
\]
where $<_L$ orders the intervals by their left endpoints. The intervals $I_1$ and $I_M$ are \emph{extreme}, while $I_j$ for $1<j<M$ is \emph{interior}. An interval rule is a map
\[
f:\mathcal D^n\to\mathcal I.
\]
Two length-$\ell$ intervals differ in a single element if and only if they are adjacent in the chain. Thus, responsiveness implies
\begin{equation}\label{eq:consec}
I_j\succ_i^R I_{j+1}
\quad\Longleftrightarrow\quad
x_jP_ix_{j+\ell}.
\end{equation}

In particular, the comparison of adjacent intervals is independent of the particular responsive extension used.

Since all intervals have the same cardinality $\ell$, the leximax extension of Definition~\ref{def:leximaxmin} is well defined on $\mathcal I$: it compares intervals according to their $P_i$-best elements, then their second-best elements, and so on. In the second subsection, we use this same leximax construction for every agent, yielding the fixed responsive extension specified there.

\begin{axiom}[Unanimity]\label{ax:unan}
An interval rule $f$ is \textbf{unanimous} if
\[
\tau(P_i)=x_1\quad\text{for all }i\in N
\quad\Longrightarrow\quad
f(P)=I_1,
\]
and
\[
\tau(P_i)=x_m\quad\text{for all }i\in N
\quad\Longrightarrow\quad
f(P)=I_M.
\]
\end{axiom}

\subsection{Responsive dominance}

In this subsection, the operative preference is responsive dominance over $\mathcal I$: that is, $\succ_i^{*}=\succ_i^R$, with weak counterpart $\succeq_i^R$.

\begin{lemma}\label{lem:interior_not_dominant}
Let $P_i\in\mathcal D$ and let $I_j$ be an interior interval. If $\tau(P_i)=x_1$, then
\[
I_j\not\succeq_i^R I_1,
\]
and if $\tau(P_i)=x_m$, then
\[
I_j\not\succeq_i^R I_M.
\]
Hence, no interior interval is an agent's most preferred outcome under every responsive extension.
\end{lemma}

\begin{proof}
Suppose $\tau(P_i)=x_1$, and consider the leximax extension $\succ_i^{\mathrm{lmax}}$ of $P_i$. Since $x_1\in I_1$ is the $P_i$-best element of $I_1$, while $x_1\notin I_j$, the $P_i$-best element of $I_j$ is strictly worse than $x_1$. Hence
\[
I_1\succ_i^{\mathrm{lmax}} I_j.
\]
Because $\succ_i^{\mathrm{lmax}}$ is a responsive extension under which $I_1$ is preferred to $I_j$, it follows that
\[
I_j\not\succ_i^R I_1.
\]
Since $I_j\neq I_1$, this implies $I_j\not\succeq_i^R I_1$. The case $\tau(P_i)=x_m$ is symmetric, using the leximax extension under which $I_M$, which contains $x_m$, is ranked above every interior interval.
\end{proof}

Lemma~\ref{lem:interior_not_dominant} shows that an interior interval can never be an agent's top outcome under every responsive extension. In the point-valued model, this observation can be combined with strategy-proofness to restrict the range to $\{I_1,I_M\}$: a strategy-proof rule selects only outcomes that can be top-ranked, while no interior interval is a top alternative. Under set-valued choice with responsive dominance, this argument no longer applies. Because responsive dominance is generally incomplete, an outcome can be immune to profitable deviations without being the top outcome under every responsive extension. The following example illustrates this possibility.

\begin{example}\label{ex:interior}
Let $m=4$, $\ell=2$, and $n=2$, so that
\[
\mathcal I=\{A,B,C\},
\qquad
A=\{x_1,x_2\},\quad
B=\{x_2,x_3\},\quad
C=\{x_3,x_4\},
\]
with
\[
A<_L B<_L C.
\]
Thus, $B$ is the unique interior interval. Consider the preferences
\[
P_1:x_1P_1x_2P_1x_4P_1x_3
\qquad\text{and}\qquad
P_2:x_4P_2x_3P_2x_1P_2x_2.
\]
Responsiveness implies
\[
A\succ_1^R B\succ_1^R C
\qquad\text{and}\qquad
C\succ_2^R B\succ_2^R A.
\]
Thus, no interval is strictly preferred to $B$ by both agents, and $B$ is Pareto efficient at this profile. Moreover, suppose that a unilateral deviation by agent~1 can change the outcome only between $B$ and $C$, while a unilateral deviation by agent~2 can change it only between $A$ and $B$. Since
\[
B\succ_1^R C
\qquad\text{and}\qquad
B\succ_2^R A,
\]
neither agent can profitably deviate from $B$. Hence $B$ can be selected without violating Pareto efficiency or strategy-proofness at this profile.

The point-valued argument excluding interior outcomes therefore fails for set-valued choice: an interior outcome need not be a top outcome for any agent in order to be Pareto efficient and immune to unilateral manipulation. Here, the relevant induced preferences happen to be complete; the distinction from a fixed responsive extension arises precisely because responsive dominance need not be complete, as illustrated by Example~\ref{ex:leximaxmin}.
\end{example}

Under responsive dominance, the induced relation on $\mathcal I$ is generally only a \emph{partial order}. For adjacent intervals, responsiveness directly determines the comparison:
\[
I_j\succ_i^R I_{j+1}
\quad\Longleftrightarrow\quad
x_jP_ix_{j+\ell},
\]
as in \eqref{eq:consec}. Comparisons between nonadjacent intervals may then follow by transitivity whenever the corresponding sequence of adjacent comparisons points in the same direction. However, when an agent's dip lies at an interior interval $I_k$, the direction of these adjacent comparisons reverses at $I_k$ (Lemma~\ref{lem:induced}). In particular, both extremes can dominate $I_k$,
\[
I_1\succ_i^R I_k
\qquad\text{and}\qquad
I_M\succ_i^R I_k,
\]
while there need be no responsive-dominance comparison between $I_1$ and $I_M$. Indeed, different responsive extensions can rank the two extremes differently, as in Example~\ref{ex:leximaxmin}. Thus, responsive dominance need not provide a complete preference over the chain of intervals.

The range restriction for single-dipped preferences therefore cannot be applied directly under responsive dominance, since it relies on a complete preference over the ordered outcomes. However, when a fixed responsive extension is selected, the resulting preference is complete and, for the leximax extension considered below, depends only on $P_i$. The induced preference over the chain
\[
I_1<_L\cdots<_L I_M
\]
can then be treated as a single-dipped preference over the ordered set of intervals, so that the corresponding range restriction applies and Lemma~\ref{lem:chain} can be invoked below.

\subsection{A fixed extension}

We now fix the common leximax extension defined above. The same extension is applied to every agent, so each agent's induced preference over $\mathcal I$ is determined solely by her reported preference $P_i$. Strategy-proofness is therefore understood with respect to this induced preference. In particular, agents with identical reports are assigned identical induced preferences, so anonymity is well defined.

\begin{lemma}\label{lem:three_point}
Let $P_i\in\mathcal D$ and let $u<v<w$. Then $v$ is not $P_i$-preferred to both $u$ and $w$.
\end{lemma}

\begin{proof}
If $v\leq d(P_i)$, then $u<v\leq d(P_i)$ and \eqref{eq:left} gives $uP_iv$. If $v>d(P_i)$, then $d(P_i)\leq v<w$ and \eqref{eq:right} gives $wP_iv$.
\end{proof}

\begin{lemma}[Induced single-dippedness]\label{lem:induced}
For every $\ell$ and every responsive extension of $P_i\in\mathcal D$, the induced preference over $(I_1,\ldots,I_M)$ is single-dipped with respect to the chain order $<_L$. Moreover, its worst interval contains $d(P_i)$, and its most preferred interval is either $I_1$ or $I_M$.
\end{lemma}

\begin{proof}
First, no interior interval can be a local maximum in the induced preference. Suppose, to the contrary, that for some $1<s<M$,

$$
I_s\succ_i I_{s-1}
\qquad\text{and}\qquad
I_s\succ_i I_{s+1}.
$$

By \eqref{eq:consec},

$$
x_{s+\ell-1}P_ix_{s-1}
\qquad\text{and}\qquad
x_sP_ix_{s+\ell}.
$$

Set

$$
a=x_{s-1},\qquad
b=x_s,\qquad
c=x_{s+\ell-1},\qquad
e=x_{s+\ell}.
$$

Then $a<b<c<e$. By Lemma~\ref{lem:three_point}, applied to $a<c<e$, $eP_ic,$ while, applied to $a<b<e$, $aP_ib.$

Together with the two comparisons above, this gives $aP_ibP_ieP_icP_ia,$ a contradiction. Now consider the consecutive comparisons along the chain. Call $I_t\succ_i I_{t+1} $ a down-step and $I_{t+1}\succ_i I_t$ an up-step. An up-step followed by a down-step would make the intervening interval a local maximum, which we have just ruled out. Hence all down-steps precede all up-steps. The induced preference therefore decreases along the chain up to a single worst interval and increases thereafter; that is, it is single-dipped with respect to $<_L$.

It remains to locate the worst interval. If $I_t$ lies entirely to the left of the dip, so that $t+\ell-1<d(P_i),$ then $x_t<x_{t+\ell}\leq d(P_i),$ and \eqref{eq:left} gives $x_tP_ix_{t+\ell}.$

By \eqref{eq:consec}, $I_t\succ_i I_{t+1}.$ Thus, intervals lying entirely to the left of the dip become worse as their location moves toward the dip. Symmetrically, if $I_t$ lies entirely to the right of the dip, then \eqref{eq:right} implies that moving toward the dip also makes the interval worse. Consequently, the worst interval cannot lie entirely on either side of $d(P_i)$ and must contain the dip.

Finally, since the induced preference is single-dipped on the chain, its most preferred interval is an extreme interval, $I_1$ or $I_M$.
\end{proof}

Write

$$
\rho(P_i)\in\{I_1,I_M\}
$$

for the most preferred interval under the induced leximax preference. Under leximax, $\rho(P_i)$ is determined by the peak of $P_i$. If $\tau(P_i)=x_1$, then $x_1$ is the $P_i$-best alternative and belongs to $I_1$. Hence $x_1$ is the leximax-best element of $I_1$, while every other interval has a strictly worse best element. Thus,

$$
\rho(P_i)=I_1.
$$

The case $\tau(P_i)=x_m$ is symmetric. Therefore,

$$
\rho(P_i)=I_1
\quad\Longleftrightarrow\quad
\tau(P_i)=x_1,
\qquad
\rho(P_i)=I_M
\quad\Longleftrightarrow\quad
\tau(P_i)=x_m.
$$

Thus, the induced top interval is consistent with the agent's top alternative. In particular, $\rho(P_i)=I_M$ for  the $n_m(P)$ agents whose top alternative is $x_m$, where $n_m(P)$ is defined in Section~\ref{sec:model}.

\begin{lemma}[Induced domain]\label{lem:induced_domain}
Under the common leximax extension, a strict linear order on $\{I_1,\ldots,I_M\}$ is induced by some $P_i\in\mathcal D$ if and only if it is single-dipped with respect to the chain order $<_L$.
\end{lemma}

\begin{proof}
($\Rightarrow$) By Lemma~\ref{lem:induced}, every order induced by some $P_i\in\mathcal D$ under leximax is single-dipped with respect to $<_L$.

($\Leftarrow$) Let $\succ$ be a single-dipped order on $I_1<_L\cdots<_L I_M$, with dip $I_k$. Its shape is determined by the directions of the adjacent comparisons: $I_j\succ I_{j+1}$ for $j<k$ and $I_{j+1}\succ I_j$ for $j\ge k$. By \eqref{eq:consec}, the induced order $\succ_i$ satisfies $I_j\succ_i I_{j+1}\iff x_jP_ix_{j+\ell}$, so we may reproduce all of these directions by choosing, for each $j$, whether $x_jP_ix_{j+\ell}$ or $x_{j+\ell}P_ix_j$; a single-dipped $P_i$ with dip at $x_k$ (say) does so, and its induced order is then single-dipped with the same dip $I_k$ as $\succ$ (Lemma~\ref{lem:induced}). The same-side comparisons of $\succ_i$ then agree with those of $\succ$. 

It remains to match the opposite-side comparisons. Under leximax, two intervals lying on opposite sides of $I_k$ are ranked by the $P_i$-comparison of their best elements, which lie on opposite sides of $d(P_i)$; and single-dippedness imposes no restriction on comparisons between alternatives on opposite sides of $d(P_i)$. These best elements can be chosen to be distinct across intervals, so their $P_i$-order is a single free linear order that leximax transmits to the opposite-side interval comparisons. Fixing that free order to reproduce $\succ$ on the opposite-side pairs yields $\succ_i=\succ$. Hence $\succ$ is induced by some $P_i\in\mathcal D$.
\end{proof}

\begin{remark}\label{rem:leximax}
The interval results use only two properties of the common leximax extension. First, the induced most preferred interval contains the top alternative, so that

$$
\rho(P_i)=I_1\iff\tau(P_i)=x_1
$$

and

$$
\rho(P_i)=I_M\iff\tau(P_i)=x_m.
$$

We refer to this as \emph{peak consistency}; it is used in Theorem~\ref{thm:interval}(ii). Second, the induced preferences exhaust the single-dipped orders on the chain, as established in Lemma~\ref{lem:induced_domain}; this is used in Theorem~\ref{thm:interval}(i). Any responsive extension satisfying these properties would yield the same results.

Not every responsive extension has these properties. In particular, under the leximin extension, which ranks sets according to their worst element, a peak-$x_1$ preference whose dip lies in $I_1$ need not induce $I_1$ as the most preferred interval. 

For example, with $m=4$, consider $P_i:x_1P_ix_4P_ix_3P_ix_2.$ The interval $I_1=\{x_1,x_2\}$ contains both the best and the worst alternatives under $P_i$, while $I_M=\{x_3,x_4\}$ has worst element $x_3$. Thus, leximin ranks $I_M$ above $I_1$, so peak consistency fails. Consequently, the reduction in Theorem~\ref{thm:interval}(ii) to the number $n_m(P)$ of agents with top alternative $x_m$ does not go through. These two properties provide us with the interval characterization as the next theorem shows.
\end{remark}

\begin{theorem}[Interval rules]\label{thm:interval}
Suppose each agent evaluates the intervals through a common responsive extension
that is \emph{peak-consistent} and whose induced preferences
over the chain of intervals $I_1<_L\cdots<_L I_M$ form the full single-dipped domain on that
chain.  
\begin{enumerate}[label=(\roman*)]
\item If $f$ is unanimous and strategy-proof, then $\operatorname{range}(f)\subseteq\{I_1,I_M\}$.
\item $f$ is anonymous, unanimous, and strategy-proof if and only if there exists a threshold $q\in\{1,\ldots,n\}$ such that
\[
f(P)=I_M \ \text{ if } n_m(P)\ge q,
\qquad
f(P)=I_1 \ \text{ if } n_m(P)<q.
\]
\end{enumerate}
\end{theorem}

\begin{proof}
(i) On a finite line with the full single-dipped domain, every unanimous strategy-proof rule selects one of the two extreme alternatives (\cite{barbera2012}); this is the finite counterpart of the continuum result of \cite{manjunath2014}. Since the agents' induced preferences over $I_1 <_L \cdots <_L I_M$ form the full single-dipped domain on this chain, whose extremes are $I_1$ and $I_M$, the restriction gives $\mathrm{range}(f) \subseteq \{I_1, I_M\}$.

(ii) By part~(i) the range is contained in the two-element chain $I_1<_L I_M$. For
these two outcomes, an agent's preference is determined by her most preferred
interval $\rho(P_i)\in\{I_1,I_M\}$, so $f$ depends on the profile only through the
induced types
\[
t_i=\begin{cases}\downarrow, & \rho(P_i)=I_1,\\ \uparrow, & \rho(P_i)=I_M.\end{cases}
\]
By Lemma~\ref{lem:chain}(i), strategy-proofness is equivalent to monotonicity, so
switching an agent's type from $\downarrow$ to $\uparrow$ can only weakly raise the
outcome in $<_L$. By anonymity, $f$ depends on the type profile only through the
number of type-$\uparrow$ agents; by peak-consistency this number equals $n_m(P)$,
so $f$ is non-decreasing in $n_m(P)$. When $n_m(P)=0$ every peak is $x_1$, so
unanimity gives $f(P)=I_1$; when $n_m(P)=n$ every peak is $x_m$, so unanimity gives
$f(P)=I_M$. Since $f$ is non-decreasing in the integer $n_m(P)$, there is a threshold
$q\in\{1,\ldots,n\}$ with $f(P)=I_M$ if $n_m(P)\ge q$ and $f(P)=I_1$ otherwise.
Conversely, every threshold rule of this form is anonymous, unanimous, and monotone
in $n_m(P)$, hence strategy-proof by Lemma~\ref{lem:chain}(i).
\end{proof}

\begin{example}[Threshold interval rules]\label{ex:threshold}
Fix the leximax extension, so $\rho(P_i)$ is the extreme interval on the side of agent $i$'s peak, and let $n_m(P)$ denote the number of agents whose peak is $x_m$. The anonymous, unanimous, strategy-proof interval rules are precisely the threshold rules

$$
f(P)=I_M\quad\text{if }n_m(P)\geq q,
\qquad
f(P)=I_1\quad\text{otherwise},
$$

one for each threshold $q\in \{1,\ldots,n\}$. At $q=1$, the rule selects $I_M$ whenever at least one agent peaks at $x_m$; at $q=n$, it selects $I_M$ only when every agent peaks at $x_m$. The majority threshold

$$
q=\left\lceil\frac{n+1}{2}\right\rceil
$$

selects the extreme interval on the side favored by a strict majority of agents. Every such rule is unanimous: when all agents peak at $x_1$, it selects $I_1$, and when all agents peak at $x_m$, it selects $I_M$. In particular, no such rule selects an interior interval.
\end{example}

\begin{remark}\label{rem:gap}
Anonymity alone does not justify reducing a profile to the count $n_m(P)$. Two single-dipped preferences may have the same most preferred interval while differing in their dip and hence in their ranking of the full chain

$$
I_1<_L\cdots<_L I_M.
$$

An anonymous rule may therefore respond to these differences. The reduction to $n_m(P)$ becomes valid only after Theorem~\ref{thm:interval}(i) restricts the range to the two extreme intervals, because on a two-element range the agent's most preferred interval determines her comparison of the two feasible outcomes. On the full chain, by contrast, the interior intervals and their relative rankings cannot in general be summarized by the peak count alone.
\end{remark}

\begin{remark}\label{rem:dominance_special}
The interior interval in Example~\ref{ex:interior} arises specifically from responsive dominance, under which an outcome is strictly preferred only when it is preferred under every responsive extension. If instead one requires that no agent gain under \emph{any} responsive extension of her preference, then in particular no agent can gain under the fixed leximax extension. Theorem~\ref{thm:interval}(i) therefore applies, and the range is contained in

$$
\{I_1,I_M\}.
$$

Thus, the additional set-valued freedom in the responsive-dominance benchmark comes from the incomparabilities permitted by responsive dominance, rather than from set-valuedness per se.
\end{remark}

\section{Set valued social choice}\label{sec:faithful}

Let the outcome be any nonempty subset of $X$, and compare subsets by responsive dominance $\succeq_i^R$. A set rule is a map
\[
G:\mathcal D^n\to 2^X\setminus\{\emptyset\},
\]
with strategy-proofness and anonymity understood with respect to responsive dominance, as in Axioms~\ref{ax:sp} and~\ref{ax:anon}. We impose one additional requirement.

\begin{axiom}[Top-selection]\label{ax:top}
$G$ satisfies \textbf{top-selection} if
\[
G(P)\subseteq\{\tau(P_i):i\in N\}.
\]
Thus, every selected alternative must be the most preferred alternative of at least one agent.
\end{axiom}

Top-selection is an efficiency requirement of a familiar kind: it rules out the selection of an alternative that no agent ranks first. On the single-dipped domain, every peak is an endpoint, and hence top-selection implies
\[
G(P)\subseteq\{x_1,x_m\}.
\]
Since outcomes are nonempty,
\[
G(P)\in
\bigl\{\{x_1\},\{x_1,x_m\},\{x_m\}\bigr\}.
\]
Top-selection also implies unanimity. Indeed, if every agent has peak $x_1$, then $x_1$ is the only available top alternative, and non-emptiness implies
\[
G(P)=\{x_1\}.
\]
The case in which all agents have peak $x_m$ is symmetric.

Only the location of the peaks at the endpoints is used in what follows, rather than the full single-dipped structure. We therefore state the characterization on the endpoint-peaked domain
\[
\mathcal E=\{P_i:\tau(P_i)\in\{x_1,x_m\}\}\supseteq\mathcal D,
\]
and compare $\mathcal E$ with the single-dipped domain in Proposition~\ref{prop:maximal}.

\begin{remark}[Efficiency does not suffice]\label{rem:eff}
Pareto efficiency cannot replace top-selection, because an interior alternative can be Pareto efficient. Let $m=4$ and consider
\[
P_1:x_1P_1x_2P_1x_3P_1x_4,
\qquad
P_2:x_4P_2x_3P_2x_2P_2x_1.
\]
Agent~2 prefers $x_2$ to $x_1$, while agent~1 prefers $x_2$ to both $x_3$ and $x_4$. Thus, no alternative is strictly preferred to $x_2$ by both agents, so $x_2$ is Pareto efficient. Efficiency therefore leaves interior alternatives eligible, whereas top-selection excludes $x_2$, since it is not the peak of any agent at this profile.
\end{remark}

The three possible outcomes form the chain
\[
\{x_1\}<_L\{x_1,x_m\}<_L\{x_m\}.
\]
On this chain, responsive dominance is complete for every endpoint-peaked preference, and the induced preference is determined solely by the location of the peak. If $\tau(P_i)=x_1$, then
\[
\{x_1\}\succ_i^R\{x_1,x_m\}\succ_i^R\{x_m\},
\]
whereas if $\tau(P_i)=x_m$, then
\[
\{x_m\}\succ_i^R\{x_1,x_m\}\succ_i^R\{x_1\}.
\]
Indeed, when the peak is $x_1$, adding $x_m$ to $\{x_1\}$ worsens the set by G\"ardenfors monotonicity, while removing the preferred alternative $x_1$ from $\{x_1,x_m\}$ leaves the singleton $\{x_m\}$ strictly worse. The case of a peak at $x_m$ is symmetric. Thus, the two-element set is every agent's second-best outcome, and the ranking of the three feasible sets depends only on the peak.

Because responsive dominance is complete on this three-element range, the same comparisons obtain under every responsive extension. Hence, imposing the stronger requirement that an agent gain under no responsive extension would not change any of the results below.

\begin{lemma}[Peak-only]\label{lem:peakonly}
If $G$ satisfies top-selection and is strategy-proof on $\mathcal E^n$, then $G(P)$ depends only on the peak profile
\[
(\tau(P_1),\ldots,\tau(P_n)).
\]
\end{lemma}

\begin{proof}
Fix an agent $i$ and the reports $P_{-i}$ of the other agents, and let $P_i$ and $P_i'$ be two preferences with the same peak. By the preceding observation, these two preferences induce the same ranking of the three feasible outcomes
\[
\{x_1\},\quad \{x_1,x_m\},\quad \{x_m\}.
\]
Suppose that
\[
G(P_i,P_{-i})\neq G(P_i',P_{-i}).
\]
Since the two induced preferences are identical and complete on the feasible outcomes, one of these two outcomes is strictly preferred to the other. The report yielding the less preferred outcome is therefore a profitable deviation to the other report, contradicting strategy-proofness. Hence
\[
G(P_i,P_{-i})=G(P_i',P_{-i}).
\]
Thus, $G$ depends on $P_i$ only through $\tau(P_i)$, and consequently on $P$ only through the peak profile.
\end{proof}

By Lemma~\ref{lem:peakonly}, we may write the outcome as a function of the coalition
\[
N_1(P)=\{i:\tau(P_i)=x_1\},
\]
with
\[
N_m(P)=N\setminus N_1(P)
\]
the coalition of agents whose peak is $x_m$. A \textbf{committee} is a nonempty family $\mathcal W\subseteq2^N$ that is \emph{monotone}: if $C\in\mathcal W$ and $C\subseteq C'$, then $C'\in\mathcal W$. A committee is the winning-coalition family of a monotone binary vote. The theorem below represents the selection of each endpoint by such a committee, and hence gives a voting-by-committees characterization in the sense of \cite{barbera1991}, specialized to the two endpoint alternatives $x_1$ and $x_m$. The covering condition guarantees that at least one endpoint is selected, while the boundary conditions ensure top-selection. We now formally define the rule.

\begin{definition}[Voting by Committees Rule]
Let $N=\{1,\dots,n\}$ and $X=\{x_1,x_m\}$ be the two endpoints. 
A set rule $G:\mathcal{E}^n \to 2^X\setminus\{\emptyset\}$ is a 
\textbf{voting by committees rule} if there exist two monotone committees 
$\mathcal{W}_1,\mathcal{W}_m \subseteq 2^N$ satisfying:
\begin{enumerate}
    \item[(i)] \textbf{Boundary conditions:} $\emptyset \notin \mathcal{W}_1$ and $\emptyset \notin \mathcal{W}_m$.
    \item[(ii)] \textbf{Covering condition:} For every coalition $C\subseteq N$,
    \[
    C\in \mathcal{W}_1 \quad \text{or} \quad N\setminus C\in \mathcal{W}_m,
    \]
    such that for every profile $P\in\mathcal{E}^n$,
    \[
    x_1\in G(P) \iff N_1(P)\in\mathcal{W}_1,
    \qquad
    x_m\in G(P) \iff N_m(P)\in\mathcal{W}_m,
    \]
    where $N_1(P)=\{i:\tau(P_i)=x_1\}$ and $N_m(P)=N\setminus N_1(P)$.\end{enumerate}
\end{definition}

We provide an example of this set rule. 

\begin{example}[A Non-Anonymous Rule]
Let $N=\{1,2,3\}$. Define committees:
\[
\mathcal{W}_1=\{C\subseteq N:1\in C\},\qquad 
\mathcal{W}_m=\{C\subseteq N:C\neq\emptyset\}.
\]
Both satisfy $\emptyset\notin \mathcal{W}_1,\emptyset\notin \mathcal{W}_m$. 
For any $C\subseteq N$, if $1\notin C$, then $1\in N\setminus C$, so $N\setminus C\in\mathcal{W}_m$; hence the covering condition holds. 
The resulting rule selects:
\[
x_1\in G(P)\iff 1\in N_1(P),\qquad 
x_m\in G(P)\iff N_m(P)\neq\emptyset.
\]
Thus, $x_1$ is chosen  when voter 1 peaks at $x_1$; $x_m$ is chosen whenever at least one voter peaks at $x_m$.
\end{example}

\begin{theorem}[Characterization]\label{thm:committee}
On the endpoint-peaked domain $\mathcal{E}^n$, a set rule 
$G:\mathcal{E}^n \to 2^X\setminus\{\emptyset\}$ satisfies 
\textit{top-selection} and is \textit{strategy-proof} 
if and only if it is a \textit{voting by committees rule} 
in the sense of Definition~1.
\end{theorem}

\begin{proof}
By Lemma~\ref{lem:peakonly}, $G(P)$ depends on $P$ only through $N_1(P)$. Write $G(C)$ for the common value of $G$ on all profiles $P$ with $N_1(P)=C$; thus $G$ is a function of the coalition $C\subseteq N$, with values in the chain
\[
\{x_1\}<_L\{x_1,x_m\}<_L\{x_m\}.
\]
An agent with peak $x_1$ has type $\downarrow$ in the sense of Lemma~\ref{lem:chain}, while an agent with peak $x_m$ has type $\uparrow$; enlarging $C$ therefore adds type-$\downarrow$ agents. By Lemma~\ref{lem:chain}(i), strategy-proofness is equivalent to monotonicity, so enlarging $C$ can only weakly decrease the outcome in $<_L$. Write
\[
A(C)=\mathbf{1}[x_1\in G(C)]\quad\text{and}\quad B(C)=\mathbf{1}[x_m\in G(C)].
\]
Because $\mathcal E$ contains preferences peaking at each endpoint, every coalition $C\subseteq N$ arises as $N_1(P)$ for some profile $P\in\mathcal E^n$: assign peak $x_1$ to the agents in $C$ and peak $x_m$ to the rest. Non-emptiness of $G$ at that profile gives
\[
A(C)\vee B(C)=1,
\]
and since every $C\subseteq N$ so arises, this holds for all $C\subseteq N$. Along the chain,
\[
(A,B)=(1,0)\ \text{for}\ \{x_1\},\qquad
(A,B)=(1,1)\ \text{for}\ \{x_1,x_m\},\qquad
(A,B)=(0,1)\ \text{for}\ \{x_m\}.
\]
Thus, as the outcome rises in $<_L$, $A$ weakly decreases and $B$ weakly increases. Since enlarging $C$ can only weakly lower the outcome, $A$ is weakly increasing in $C$ and $B$ is weakly decreasing in $C$. Define
\[
\mathcal W_1=\{C\subseteq N:A(C)=1\}
\qquad\text{and}\qquad
\mathcal W_m=\{C\subseteq N:B(N\setminus C)=1\}.
\]
The monotonicity of $A$ and $B$ implies that both $\mathcal W_1$ and $\mathcal W_m$ are committees. Moreover, $A(C)\vee B(C)=1$ for every $C\subseteq N$ is precisely the covering condition
\[
C\in\mathcal W_1
\quad\text{or}\quad
N\setminus C\in\mathcal W_m.
\]
Finally, top-selection implies the boundary conditions. If $N_1(P)=\emptyset$, no agent has peak $x_1$, so top-selection gives $x_1\notin G(P)$; hence $\emptyset\notin\mathcal W_1$. Similarly, if $N_m(P)=\emptyset$, top-selection gives $x_m\notin G(P)$, so $\emptyset\notin\mathcal W_m$.

Conversely, suppose that $\mathcal W_1$ and $\mathcal W_m$ satisfy the committee, boundary, and covering conditions, and define $G$ by
\[
x_1\in G(P)\iff N_1(P)\in\mathcal W_1,
\qquad
x_m\in G(P)\iff N_m(P)\in\mathcal W_m.
\]
The covering condition guarantees that $G(P)$ is nonempty. Since $\mathcal W_1$ is a committee, membership of $N_1(P)$ in $\mathcal W_1$ is preserved when $N_1(P)$ is enlarged; similarly, membership of $N_m(P)$ in $\mathcal W_m$ is preserved when $N_m(P)$ is enlarged. Hence $A$ is non-decreasing and $B$ is non-increasing in $C$, so $G$ is monotone on the chain, and Lemma~\ref{lem:chain}(i) gives strategy-proofness.

For top-selection, if $x_1\in G(P)$, then $N_1(P)\in\mathcal W_1$. Since $\emptyset\notin\mathcal W_1$, we have $N_1(P)\neq\emptyset$, so some agent has peak $x_1$. The argument for $x_m$ is symmetric. Thus every selected alternative is the peak of some agent, and $G$ satisfies top-selection.
\end{proof}

\begin{corollary}[Anonymous rules]\label{cor:quota}
A set rule
\[
G:\mathcal E^n\to2^X\setminus\{\emptyset\}
\]
is anonymous, top-selecting, and strategy-proof if and only if there exist integers
\[
1\leq q_1\leq q_2\leq n
\]
such that
\[
G(P)=
\begin{cases}
\{x_1\}, & n_m(P)<q_1,\\
\{x_1,x_m\}, & q_1\leq n_m(P)<q_2,\\
\{x_m\}, & n_m(P)\geq q_2.
\end{cases}
\]
\end{corollary}

\begin{proof}
Anonymity implies that the two committees in Theorem~\ref{thm:committee} depend only on coalition cardinality. Hence, each is a quota family. Since
\[
n_m(P)=|N_m(P)|
\qquad\text{and}\qquad
|N_1(P)|=n-n_m(P),
\]
there are thresholds at which $x_m$ becomes selected and $x_1$ ceases to be selected. The boundary conditions imply that neither endpoint is selected when its supporting coalition is empty, while the covering condition ensures that at least one endpoint is selected at every profile. Thus, for some
\[
1\leq q_1\leq q_2\leq n,
\]
$x_m$ is selected if and only if $n_m(P)\geq q_1$, while $x_1$ is selected if and only if $n_m(P)<q_2$. This gives the stated form.
\end{proof}

\begin{example}[Committee rules with $n=3$]\label{ex:committee}
Let $N=\{1,2,3\}$.

\emph{(a) An anonymous rule.}
Take thresholds $q_1=1$ and $q_2=3$. Equivalently,
\[
\mathcal W_1=\{C:|C|\geq1\},
\qquad
\mathcal W_m=\{C:|C|\geq1\}.
\]
Thus, an endpoint is selected whenever at least one agent has that endpoint as her peak. If all three agents peak at $x_1$, then
\[
G(P)=\{x_1\};
\]
if all three peak at $x_m$, then
\[
G(P)=\{x_m\};
\]
and whenever the agents disagree,
\[
G(P)=\{x_1,x_m\}.
\]
This is a ``consensus-or-compromise'' rule: it selects a singleton endpoint under unanimous support and the two-element set whenever the peaks are split.

\emph{(b) A non-anonymous rule.}
Take
\[
\mathcal W_1=\{C:1\in C\}
\qquad\text{and}\qquad
\mathcal W_m=\{C:C\neq\emptyset\}.
\]
Both are committees, and
\[
\emptyset\notin\mathcal W_1,
\qquad
\emptyset\notin\mathcal W_m.
\]
The covering condition also holds: if $1\notin C$, then $1\in N\setminus C$, so $N\setminus C\neq\emptyset$ and hence $N\setminus C\in\mathcal W_m$. Here $x_1$ is selected if and only if agent~1's peak is $x_1$, while $x_m$ is selected if and only if at least one agent has peak $x_m$. Thus, if agent~1 peaks at $x_1$ and at least one other agent peaks at $x_m$,
\[
G(P)=\{x_1,x_m\};
\]
if agent~1 peaks at $x_m$, then
\[
G(P)=\{x_m\}
\]
regardless of the other agents' peaks. Agent~1 is therefore decisive for the selection of $x_1$, but has no veto over $x_m$. Strategy-proofness follows from the monotonicity of the induced rule on the chain.
\end{example}

\begin{proposition}[Maximal domain]\label{prop:maximal}
The endpoint-peaked domain $\mathcal E$ is maximal, under set inclusion, for the peak-based characterization of Theorem~\ref{thm:committee}. In particular, the same characterization applies on every subdomain of $\mathcal E$, whereas if a domain contains a preference with an interior peak, top-selection forces an outcome outside the set $\bigl\{\{x_{1}\},\{x_{1},x_{m}\},\{x_{m}\}\bigr\}$, so no representation of the form in Theorem~\ref{thm:committee} exists. Thus, single-dippedness is sufficient but not necessary: the characterization uses only the location of peaks at the endpoints.
\end{proposition}

\begin{proof}
Let $\mathcal D'\subseteq\mathcal E$. Every preference in $\mathcal D'$ has its peak at one of the two endpoints. Consequently, top-selection restricts the range to
\[
\bigl\{\{x_1\},\{x_1,x_m\},\{x_m\}\bigr\},
\]
and responsive dominance over these three outcomes is determined solely by the location of the peak,  as above. Hence the proof of Lemma~\ref{lem:peakonly} and the subsequent peak-based characterization apply on $(\mathcal D')^n$.

Conversely, suppose a domain contains a preference $P^*$ with an interior peak
\[
\tau(P^*)=x_j,
\qquad
1<j<m.
\]
Consider the profile at which every agent reports $P^*$. Top-selection then requires
\[
G(P)\subseteq\{x_j\}.
\]
Since outcomes are nonempty,
\[
G(P)=\{x_j\}.
\]
This outcome is not among
\[
\bigl\{\{x_1\},\{x_1,x_m\},\{x_m\}\bigr\},
\]
so no representation of the form in Theorem~\ref{thm:committee} can describe a top-selecting rule on such a domain. Thus, once an interior-peaked preference is admitted, the three-outcome committee characterization no longer applies.
\end{proof}

\begin{proposition}[Group-strategy-proofness]\label{prop:gsp}
Every rule characterized in Theorem~\ref{thm:committee} is group-strategy-proof.
\end{proposition}

\begin{proof}
By Theorem~\ref{thm:committee}, every such rule is a monotone rule on the chain
\[
\{x_1\}<_L\{x_1,x_m\}<_L\{x_m\}.
\]
A peak-$x_1$ agent has type $\downarrow$, while a peak-$x_m$ agent has type $\uparrow$. Lemma~\ref{lem:chain}(ii) therefore implies group-strategy-proofness. Intuitively, a deviation that moves the outcome upward in $<_L$ can benefit only agents of type $\uparrow$, while any type-$\downarrow$ member of the deviating coalition is strictly worse off; monotonicity prevents a coalition consisting only of type-$\uparrow$ agents from generating such an upward movement by their reports. The downward case is symmetric.
\end{proof}

\begin{example}[Why a coalition cannot gain]\label{ex:gsp}
Consider the anonymous rule of Example~\ref{ex:committee}(a) with $n=3$ and thresholds
\[
q_1=q_2=2.
\]
Thus,
\[
G(P)=
\begin{cases}
\{x_1\}, & n_m(P)\leq1,\\
\{x_m\}, & n_m(P)\geq2.
\end{cases}
\]
Suppose agents $1$ and $2$ peak at $x_1$, while agent $3$ peaks at $x_m$. Then $n_m(P)=1$ and
\[
G(P)=\{x_1\}.
\]
Agents $1$ and $2$ already receive their most preferred outcome, so no coalition consisting only of them can benefit from a deviation. Agent~3 would prefer $\{x_m\}$, but reaching this outcome requires $n_m(P)\geq2$, so at least one of agents $1$ and $2$ would have to report a peak at $x_m$. Such an agent would then receive $\{x_m\}$ instead of $\{x_1\}$ and would be strictly worse off under her true preference. Hence, every deviation that moves the outcome upward in $<_L$ in order to benefit agent~3 makes a peak-$x_1$ member strictly worse off. No coalition can therefore make all its members weakly better off with at least one member strictly better off, illustrating Proposition~\ref{prop:gsp}.
\end{example}

\section{Concluding remarks}\label{sec:conclusion}

On single-dipped domains, strategy-proof set-valued choice is organized by the two endpoints of the line. For fixed-length intervals compared through a common responsive extension, the induced preference over the ordered intervals is again single-dipped, so the unanimous strategy-proof rules select an extreme interval, and the anonymous rules are parameterized by a single threshold (Theorem~\ref{thm:interval}). For arbitrary sets under top-selection, the strategy-proof rules are the voting-by-committees rules, with one committee governing each endpoint (Theorem~\ref{thm:committee}) and the anonymous rules taking a two-threshold quota form (Corollary~\ref{cor:quota}): the outcome is a singleton when agents agree on the preferred endpoint and the two-element set when they do not. Pareto efficiency does not give this, since interior alternatives can be Pareto efficient; it is top-selection, requiring each chosen alternative to be some agent's peak, that leaves only the endpoints.

The endpoint restriction itself is not assumed but derived: unanimity and strategy-proofness force the range into $\{x_1,x_m\}$ on the single-dipped domain, and into $\{I_1,I_M\}$ in the interval benchmark. Whether an analogue holds for set-valued rules without fixing the domain in advance is open and less immediate. The range is now a family of sets rather than alternatives, its boundedness in any natural order on sets is not automatic, and the completeness of responsive dominance over the relevant outcomes---which is what makes the committee structure tractable in Theorem~\ref{thm:committee}---is no longer guaranteed once the three-set range is not fixed in advance. The committee objects themselves would have to be determined rather than taken as given, and it is not clear that strategy-proofness alone provides enough structure to force that.
\bibliographystyle{plainnat}
\bibliography{order}

\end{document}